\documentclass{llncs}

\usepackage[ruled, vlined, nofillcomment, linesnumbered]{algorithm2e}

\usepackage{cite}
\usepackage{url}
\usepackage{amsmath}
\usepackage{amssymb}
\usepackage{graphicx}
\usepackage[english]{babel}
\usepackage{booktabs}
\usepackage{nicefrac}
\usepackage{verbatim}
\usepackage{subcaption}
\usepackage{xspace}
\usepackage[numbers, sectionbib]{natbib}

\usepackage{tikz}
\usepackage{pgfplots}

\usepackage[pdftitle={Faster way to agony},
           pdfauthor={},
		   hidelinks,
		   hyperfootnotes = false]{hyperref}

\newcommand{\set}[1]{\left\{#1\right\}}

\newcommand{\fpr}[1]{\mathopen{}\left(#1\right)}

\newcommand{\abs}[1]{{\left|#1\right|}}

\newcommand{\np}{\textbf{NP}}
\newcommand{\apx}{\textbf{APX}}

\newcommand{\ints}{\mathbb{Z}}
\newcommand{\funcdef}[3]{{#1}:{#2} \to {#3}}
\newcommand{\define}{\leftarrow}

\DeclareRobustCommand{\dispfunc}[2]{%
  \ensuremath{%
  \ifthenelse{\equal{#2}{}}%
    {\mathit{#1}}%
    {\mathit{#1}\fpr{#2}}}}

\newcommand{\agony}[1]{\dispfunc{q}{#1}}
\newcommand{\slack}[1]{\dispfunc{slack}{#1}}
\newcommand{\parent}[1]{\dispfunc{parent}{#1}}

\newcommand{\relief}{\textsc{Relief}\xspace}
\newcommand{\minagony}{\textsc{MinAgony}\xspace}
\newcommand{\cycledfs}{\textsc{CycleDFS}\xspace}

\newcommand{\fm}[1]{\mathcal{#1}}

\newlength{\levelsep}
\newlength{\nodesep}
\tikzstyle{node} = [fill = white, circle, inner sep = 1pt]
\tikzstyle{label} = [inner sep = 0pt]

\tikzstyle{euleredge} = [yafcolor2, -latex, thick, densely dotted]
\tikzstyle{dagedge} = [yafcolor1, -latex, thick]
\tikzstyle{leveledge} = [yafaxiscolor!50, dashed]

\SetKwComment{tcpas}{\{}{\}}
\SetCommentSty{textnormal}
\SetArgSty{textnormal}
\SetKw{False}{false}
\SetKw{True}{true}
\SetKw{Null}{null}
\SetKwInOut{Output}{output}
\SetKwInOut{Input}{input}
\SetKw{AND}{and}
\SetKw{OR}{or}
\SetKw{Break}{break}

\pgfdeclarelayer{background}
\pgfdeclarelayer{foreground}
\pgfsetlayers{background,main,foreground}

\definecolor{yafaxiscolor}{rgb}{0.3, 0.3, 0.3}

\definecolor{yafcolor1}{rgb}{0.4, 0.165, 0.553}
\definecolor{yafcolor2}{rgb}{0.949, 0.482, 0.216}
\definecolor{yafcolor3}{rgb}{0.47, 0.549, 0.306}
\definecolor{yafcolor4}{rgb}{0.925, 0.165, 0.224}
\definecolor{yafcolor5}{rgb}{0.141, 0.345, 0.643}
\definecolor{yafcolor6}{rgb}{0.965, 0.933, 0.267}
\definecolor{yafcolor7}{rgb}{0.627, 0.118, 0.165}
\definecolor{yafcolor8}{rgb}{0.878, 0.475, 0.686}

\newlength{\yafaxispad}
\newlength{\yaftlpad}
\newlength{\yaflabelpad}
\newlength{\yafaxiswidth}
\newlength{\yafticklen}
\makeatletter
\def\pgfplots@drawtickgridlines@INSTALLCLIP@onorientedsurf#1{}
\makeatother

\newcommand{\yafdrawxaxis}[2]{
	\pgfplotstransformcoordinatex{#1}\let\xmincoord=\pgfmathresult 
	\pgfplotstransformcoordinatex{#2}\let\xmaxcoord=\pgfmathresult 
	\pgfsetlinewidth{\yafaxiswidth} 
	\pgfsetcolor{yafaxiscolor}
	\pgfpathmoveto{\pgfpointadd{\pgfpointadd{\pgfplotspointrelaxisxy{0}{0}}{\pgfqpointxy{\xmincoord}{0}}}{\pgfqpoint{-0.5\yafaxiswidth}{\yafaxispad}}}
	\pgfpathlineto{\pgfpointadd{\pgfpointadd{\pgfplotspointrelaxisxy{0}{0}}{\pgfqpointxy{\xmaxcoord}{0}}}{\pgfqpoint{0.5\yafaxiswidth}{\yafaxispad}}}
	\pgfusepath{stroke}

}
\newcommand{\yafdrawyaxis}[2]{
	\pgfplotstransformcoordinatey{#1}\let\ymincoord=\pgfmathresult 
	\pgfplotstransformcoordinatey{#2}\let\ymaxcoord=\pgfmathresult 
	\pgfsetlinewidth{\yafaxiswidth} 
	\pgfsetcolor{yafaxiscolor}
	\pgfpathmoveto{\pgfpointadd{\pgfpointadd{\pgfplotspointrelaxisxy{0}{0}}{\pgfqpointxy{0}{\ymincoord}}}{\pgfqpoint{\yafaxispad}{-0.5\yafaxiswidth}}}
	\pgfpathlineto{\pgfpointadd{\pgfpointadd{\pgfplotspointrelaxisxy{0}{0}}{\pgfqpointxy{0}{\ymaxcoord}}}{\pgfqpoint{\yafaxispad}{0.5\yafaxiswidth}}}
	\pgfusepath{stroke}
}

\newcommand{\yafdrawaxis}[4]{\yafdrawxaxis{#1}{#2}\yafdrawyaxis{#3}{#4}}

\pgfplotscreateplotcyclelist{yaf}{%
{yafcolor1,mark options={scale=0.75},mark=o}, 
{yafcolor2,mark options={scale=0.75},mark=square},
{yafcolor3,mark options={scale=0.75},mark=triangle},
{yafcolor4,mark options={scale=0.75},mark=o},
{yafcolor5,mark options={scale=0.75},mark=o},
{yafcolor6,mark options={scale=0.75},mark=o},
{yafcolor7,mark options={scale=0.75},mark=o},
{yafcolor8,mark options={scale=0.75},mark=o}} 

\pgfplotsset{axis y line=left, axis x line=bottom,
	tick align=outside,
	compat = 1.3,
	tickwidth=\yafticklen,
	clip = false,
	every axis title shift = 0pt,
    x axis line style= {-, line width = 0pt, opacity = 0},
    y axis line style= {-, line width = 0pt, opacity = 0},
    x tick style= {line width = \yafaxiswidth, color=yafaxiscolor, yshift = \yafaxispad},
    y tick style= {line width = \yafaxiswidth, color=yafaxiscolor, xshift = \yafaxispad},
    x tick label style = {font=\scriptsize, yshift = \yaftlpad},
    y tick label style = {font=\scriptsize, xshift = \yaftlpad},
    every axis y label/.style = {at = {(ticklabel cs:0.5)}, rotate=90, anchor=center, font=\scriptsize, yshift = -\yaflabelpad},
    every axis x label/.style = {at = {(ticklabel cs:0.5)}, anchor=center, font=\scriptsize, yshift = \yaflabelpad},
    x tick label style = {font=\scriptsize, yshift = 1pt},
    grid = major,
    major grid style  = {dash pattern = on 1pt off 3 pt},
	every axis plot post/.append style= {line width=\yafaxiswidth} ,
	legend cell align = left,
	legend style = {inner sep = 1pt, cells = {font=\scriptsize}},
	legend image code/.code={%
		\draw[mark repeat=2,mark phase=2,#1] 
		plot coordinates { (0cm,0cm) (0.15cm,0cm) (0.3cm,0cm) };%
	} 
}

\newcommand{\spara}[1]{{\smallskip\noindent{\bf {#1}}}}

\begin{document}

\title{Faster way to agony}
\subtitle{Discovering hierarchies in directed graphs}

\author{Nikolaj Tatti}
\institute{Helsinki Institute for Information Technology \\  
Department of Information and Computer Science\\
Aalto University, Finland\\
\url{nikolaj.tatti@aalto.fi}}

\maketitle

\begin{abstract}
Many real-world phenomena exhibit strong hierarchical structure.  Consequently,
in many real-world directed social networks vertices do not play equal role.
Instead, vertices form a hierarchy such that the edges appear mainly from upper
levels to lower levels. Discovering hierarchies from such graphs
is a challenging problem that has gained attention. Formally, given a directed graph, we
want to partition vertices into levels such that ideally there are only edges
from upper levels to lower levels.  From computational point of view, the ideal
case is when the underlying directed graph is acyclic. In such case, we can
partition the vertices into a hierarchy such that there are only edges from
upper levels to lower edges.  In practice, graphs are rarely acyclic, hence we
need to penalize the edges that violate the hierarchy. One practical approach
is agony, where each violating edge is penalized based on the severity of the
violation.  The fastest algorithm for computing agony requires $O(nm^2)$ time.
In the paper we present an algorithm for computing agony that has better
theoretical bound, namely $O(m^2)$.  We also show that in practice the obtained
bound is pessimistic and that we can use our algorithm to compute agony for
large datasets. Moreover, our algorithm can be used as any-time algorithm.

\end{abstract}

\section{Introduction}
\label{sec:intro}

Many real-world phenomena exhibit strong hierarchical
structure~\cite{DBLP:conf/icdm/MacchiaBGC13,clauset:08:hierarchy,DBLP:conf/cse/MaiyaB09,jameson:99:behaviour,elo1978rating}.
For example, it is
more likely that a manager in a large company will write emails to the her
subordinates than an employee writes an email to his manager. As another
example, in a tournament, it is more likely that a better team will
win a second-tear team. 

Discovering hierarchy in the context of directed networks can be viewed as the
following optimization problem. Given a directed graph, partition vertices into
levels such that there are only edges from upper levels to lower levels.  For
example, consider an email communication network of a large institute,
directed edge $x \to y$ is created if $x$ has written an email to $y$. We
should expect that the upper level of the hierarchy consists of top-level
managers and each level consists of subordinates of the previous level.

Unfortunately, such a partition is only possible when the graph does not have
cycles, a rare case in practice. Instead a more fruitful approach is to find a
hierarchy that minimizes some cost function. One possible cost function is to
penalize every edge that violates the hierarchy with a constant cost.
Unfortunately, this problem leads to \textsc{Feedback Arc Set} problem, where
we are asked to discover a maximal directed acyclic subgraph. This problem is a
classic \np-hard problem~\cite{dinur:05:cover}.

A practical variant of discovering hierarchies that was introduced recently by~\citet{gupte:11:agony} is
to weight the edges based on the severity of the violation of hierarchy.
Unlike the constant weights, this problem can be solved in $O(nm^2)$,
polynomial time, where $n$ is the number of vertices and $m$ is the number of
edges.

In this paper we introduce a new algorithm for computing a hierarchy that
minimizes agony. Our algorithm achieves computational complexity of $O(m^2)$
which is significantly better than $O(nm^2)$, the computational complexity of
the currently best approach. We also demonstrate empirically that $O(m^2)$ is
in fact pessimistic and that we can compute agony using our approach for large
networks.

Our approach is based on a primal-dual technique. Minimizing agony has an
interpretable dual problem, finding eulerian subgraph, a graph where the
in-degree is equal to the out-degree for each vertex, with the maximum number of
edges. This relation implies that the agony will always be at least as large as
any eulerian subgraph.  We are able to exploit this relation by designing an
iterative algorithm.  At each iteration we decrease the gap between the current
agony and the current eulerian subgraph by either modifying the hierarchy or
modifying the eulerian subgraph. We show that each iteration requires only
$O(m)$ time and we need at most $m$ steps.

The rest of the paper is organized as follows.  We introduce the notation and
state the optimization problem in Section~\ref{sec:prel}.  In
Section~\ref{sec:pd} we review the connection between agony and eulerian
subgraphs. In Section~\ref{sec:discovery} we introduce our optimization
algorithm. We discuss the related work in Section~\ref{sec:related} and present
experimental evaluation in Section~\ref{sec:exps}. Finally, we conclude
the paper with remarks in Section~\ref{sec:concl}.

\section{Preliminaries and problem statement}
\label{sec:prel}

Throughout the whole paper we assume that we are given a directed graph $G =
(V, E)$.  We will denote the number of vertices by $n = \abs{V}$ and the number
of edges by $m = \abs{E}$. \emph{All} graphs in this paper are directed and
have the same vertices $V$. 
Given a graph $H = (V, F)$, we will write $E(H) = F$.

In this paper our goal is to discover a hierarchy among vertices in a graph $G$.
That is, assume that we are given a graph $G = (V, E)$ and our goal is to
discover a partition of vertices $\fm{P} = P_1, \ldots, P_k$, such that $P_i
\cap P_j = \emptyset$ and $\bigcup_{i = 1}^k P_i = V$, optimizing a certain
quality score which we will define later.  It will be more convenient to
express this partition using a rank function, that is, our goal is to construct
a function $\funcdef{r}{V}{\ints}$ mapping each vertex to an integer. We can
easily construct a partition from this rank function by grouping the nodes
mapping to the same value together.

Our next step is to define the quality score.

Given a rank function $r$,
we say that an edge $e = (u, v) \in E$ is \emph{forward} if $r(u) < r(v)$.
Similarly, we say that $e = (u, v) \in E$ is \emph{backward}
if $r(u) \geq r(v)$. Note that the inequality is strict for forward edges.

As our goal is to discover hierarchy in $G$, in an ideal partition all edges
are forward.  This is only possible if $G$ is a DAG which is rarely the case in practice. 
Consequently, we need a quality score that would penalize the backward edges. 
Given a rank $r$ we define the \emph{agony} of an edge $(u, v) \in E$
to be
\[
	\agony{(u, v), r} = \max(r(u) - r(v) + 1, 0)\quad.
\]
The agony for forward edges is $0$ while the agony for backward edges is the
difference between ranks plus 1.  Note that the edges within the same block are
penalized by $1$.

Given a graph $G$ and a rank $r$ we define the agony of the whole graph to be the sum of individual edges,
\[
	\agony{G, r} = \sum_{e \in E} \agony{e, r}\quad.
\]

\begin{figure}[ht!]
\hfill
\begin{tikzpicture}

\draw[leveledge] (-0.5\nodesep, 0) -- (4.5\nodesep, 0);
\draw[leveledge] (-0.5\nodesep, \levelsep) -- (4.5\nodesep, \levelsep);
\draw[leveledge] (-0.5\nodesep, 2\levelsep) -- (4.5\nodesep, 2\levelsep);
\draw[leveledge] (-0.5\nodesep, 3\levelsep) -- (4.5\nodesep, 3\levelsep);

\node[node] (u0) at (\nodesep, 0) {$a$};
\node[node] (u1) at (0, \levelsep) {$b$};
\node[node] (u2) at (\nodesep, 2\levelsep) {$c$};
\node[node] (u3) at (0, 3\levelsep) {$d$};

\node[node] (u4) at (2\nodesep, 1\levelsep) {$e$};
\node[node] (u5) at (3\nodesep, 0\levelsep) {$f$};
\node[node] (u6) at (4\nodesep, 1\levelsep) {$g$};
\node[node] (u7) at (3\nodesep, 2\levelsep) {$h$};

\draw[euleredge] (u0) edge (u2);
\draw[euleredge] (u2) edge (u3);
\draw[euleredge] (u3) edge (u1);
\draw[euleredge] (u1) edge (u0);

\draw[euleredge] (u5) edge (u4);
\draw[euleredge] (u4) edge (u6);
\draw[euleredge] (u6) edge (u5);
\draw[dagedge] (u6) edge (u7);
\draw[dagedge] (u4) edge (u7);

\draw[dagedge] (u0) edge (u4);
\end{tikzpicture}\hfill
\begin{tikzpicture}

\draw[leveledge] (-0.5\nodesep, \levelsep) -- (4.5\nodesep, \levelsep);
\draw[leveledge] (-0.5\nodesep, 2\levelsep) -- (4.5\nodesep, 2\levelsep);
\draw[leveledge] (-0.5\nodesep, 3\levelsep) -- (4.5\nodesep, 3\levelsep);

\node[node] (u0) at (\nodesep, \levelsep) {$a$};
\node[node] (u1) at (0, \levelsep) {$b$};
\node[node] (u2) at (\nodesep, 3\levelsep) {$c$};
\node[node] (u3) at (0, 3\levelsep) {$d$};

\node[node] (u4) at (2\nodesep, 2\levelsep) {$e$};
\node[node] (u5) at (2.5\nodesep, 1\levelsep) {$f$};
\node[node] (u6) at (4\nodesep, 1\levelsep) {$g$};
\node[node] (u7) at (3\nodesep, 3\levelsep) {$h$};

\draw[dagedge] (u0) edge (u2);
\draw[euleredge] (u2) edge (u3);
\draw[euleredge] (u3) edge (u1);
\draw[euleredge] (u1) edge (u0);

\draw[euleredge] (u5) edge (u4);
\draw[euleredge] (u4) edge (u6);
\draw[euleredge] (u6) edge (u5);
\draw[dagedge] (u6) edge (u7);
\draw[dagedge] (u4) edge (u7);

\draw[euleredge] (u0) edge (u4);
\draw[euleredge] (u4) edge (u2);
\end{tikzpicture}\hspace*{\fill}

\caption{Toy graphs. Dotted edges represent the eulerian subgraph. Ranks are represented by dashed grey horizontal
lines.}
\label{fig:ex}
\end{figure}
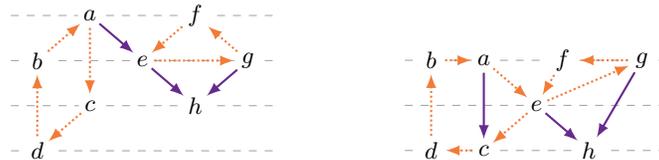

\begin{example}
The agony of the left graph given in Figure~\ref{fig:ex} is equal to
\[
	\agony{(b, a)} + \agony{(d, b)} + \agony{(e, g)} + \agony{(g, f)} = 2 + 3 + 1 + 2 = 8\quad.
\]
The agony of the right graph is equal to
\[
	\agony{(b, a)} + \agony{(d, b)} + \agony{(c, d)} + \agony{(e, g)} + \agony{(g, f)} = 1 + 3 + 1 + 2 + 1 = 8\quad.
\]
\end{example}

We can now state the main optimization problem of this paper.

\begin{problem}
Given a graph $G$ find a rank function $r$ minimizing agony $\agony{G, r}$.
\end{problem}

Graph $H = (V, F)$ is called \emph{eulerian} if the out-degree of each vertex is equal to its in-degree,
\[
	\abs{\set{u \in V ; (v, u) \in F}} = \abs{\set{w \in V ; (v, w) \in F}} \quad.
\]
In the literature, $H$ is sometimes required to be connected but here we do not impose this constraint.

\begin{example}
An example of eulerian subgraph in the left graph of Figure~\ref{fig:ex} consists of
$(b, a)$, 
$(a, c)$, 
$(c, d)$, 
$(d, b)$, 
$(f, e)$, 
$(e, g)$, and 
$(g, f)$. 

An example of eulerian subgraph in the right graph of Figure~\ref{fig:ex} consists of
$(b, a)$, 
$(a, e)$, 
$(e, c)$, 
$(c, d)$, 
$(d, b)$, 
$(f, e)$, 
$(e, g)$, and 
$(g, f)$. 
\end{example}

Given a graph $G$ we say that $H$ is a \emph{maximum} eulerian subgraph $G$ if
$H$ is an eulerian subgraph of $G$ and has the highest number of edges among
all eulerian subgraphs of $G$.  This graph is not necessarily unique.
For notational simplicity, we require that $G$ and $H$
have the same vertices, $V(H) = V(G) = V$. This restriction does not impose any
difficulties since we can always add missing vertices as singletons to $H$.

Given a graph $G$ we say that $H$ is a \emph{maximal} eulerian subgraph $G$ if
$H$ is an eulerian subgraph of $G$ and we cannot increase $H$ by adding new
edges without making it non-eulerian. Note that maximum eulerian subgraph is
necessarily maximal but not the other way around. It is easy to see that $H$
is maximal if and only if the remaining edges in $G$ form a DAG.

As we see in the next section, the following optimization problem, that is, finding
the maximum eulerian subgraph is closely related to optimizing agony.

\begin{problem}
Given a graph $G$ find an eulerian subgraph $H$ maximizing $\abs{E(H)}$, the number of edges.
\end{problem}

\section{Agony and eulerian subgraphs}\label{sec:pd}

In this section we review the connection between agony and discovering maximum
eulerian subgraph. In fact, they are dual problems. This connection allows us
to develop our algorithm in the next sections.

To see the connection let us first write the agony optimization problem as an
integer linear program, that is, our goal is to solve the following program.
\begin{align}
	\min \sum_{(u, v) \in E} p(u, v) &&& \text{such that} \label{eq:primal}\\
	p(u, v) &\geq r(v) - r(u) + 1 && \text{for all } (u, v) \in E,  \nonumber\\
	p(u, v) &\geq 0               && \text{for all } (u, v) \in E, \nonumber\\
	p(u, v),\ r(w) & \in \ints           && \text{for all } (u, v) \in E,\ w \in V\nonumber\quad. 
\end{align}
The solution for Eq.~\ref{eq:primal} will contain the optimal rank function $r$ and 
agony for individual edges $p(u, v)$.

Let us relax the program by dropping the integrality conditions, thus
transforming the program into a standard linear program. The dual of this program is
equal to
\begin{align}
	\label{eq:dual}
	\max \sum_{(u, v) \in E} c(u, v) &&& \text{such that} \\
	c(u, v) &\leq 1               && \text{for all } (u, v) \in E,  \nonumber\\
	\sum_{(u, v) \in E} c(u, v) & =  \sum_{(v, w) \in E} c(v, w)                && \text{for all } v \in V,\nonumber\\
	c(u, v) &\geq 0               && \text{for all } (u, v) \in E\quad.  \nonumber\\
\end{align}
Assume that we are given a feasible solution to a dual problem such that $c(u,
v)$ are integral.  The conditions imply that $c(u, v)$ is either $0$ or $1$.
If we form a subgraph $H$ by taking the edges for which $c(u, v) = 1$, then the
equality condition implies immediately that $H$ is eulerian. Consequently,
the solution for the dual problem is at least as large as the number of edges
in the maximum eulerian graph.

Since the primal solution is always larger than the dual solution we have the following
proposition.

\begin{proposition}
\label{prop:primaldual}
Assume that we are given a graph $G$.  Let $r$ be a rank function 
and let $H$ be an eulerian subgraph. Then $\abs{E(H)} \leq \agony{G, r}$.
Moreover, if $\abs{E(H)} = \agony{G, r}$, then $r$ minimizes agony and $H$
has the maximum number of edges.
\end{proposition}

\begin{proof}
Let $P$ be the solution of Eq.~\ref{eq:primal} and let $D$ be the solution of
Eq.~\ref{eq:dual}. Primal-dual theory (see, for example, \cite{Papadimitriou:1982:COA:31027}) states that $D = P$.
We now have $\agony{G, r} \geq P = D \geq \abs{E(H)}$.
If $\abs{E(H)} = \agony{G, r}$, then this immediately implies  $\abs{E(H)} = \agony{G, r} = P = D$,
proving the optimality of $r$ and $H$.\qed
\end{proof}

The previous result only proves that \emph{if} there is a rank function $r$
whose agony corresponds to the number of edges in the eulerian subgraph $H$,
then $r$ and $H$ are optimal. It does not guarantee that such solution exists.
\citet{gupte:11:agony} showed that such solution always exists. However, we do not need this
result. Instead, in the next section we introduce an algorithm that finds $r$
and $H$ satisfying the conditions of Proposition~\ref{prop:primaldual} which
immediately implies the optimality of $r$.

\section{Algorithm for discovering agony}\label{sec:discovery}

In this section we present our algorithm based on the results of previous
section. As our first step, we characterize the difference between the agony of
the current rank function and the number of edges in the eulerian subgraph.  We
then present an algorithm that minimizes this difference and by doing so leads
to the optimal solution. Finally, we present a fast algorithm for discovering a
maximal eulerian subgraph, an initialization step that is needed for our main
algorithm.

\subsection{Gap between agony and eulerian subgraphs}

In order to characterize the gap between the scores we need several concepts.

Assume that we are given a graph $G$ and let $H$ be a maximal eulerian subgraph
$G$.  We say that a rank function $r$ \emph{conforms} $H$ if all backward edges with
respect to $r$ are in $H$. Note that this is possible only if $H$ is maximal,
otherwise there will be at least one backward edge in $E(G) \setminus E(H)$. 

We will express the gap as a sum of slacks. More formally, given a rank $r$ we
define the \emph{slack} of an edge as
\[
	\slack{(u, v), r} = \max(r(v) - r(u) - 1, 0) \quad.
\]
Slack of $(u, v)$ will be positive only if the edge is forward and the rank $r(v)$
is at least $r(u) + 2$.

We saw in the previous section that the agony is always larger than the number
of edges in an eulerian graph. We can express this difference under certain
conditions using slacks.

\begin{proposition}
\label{prop:slack}
Assume that we are given a graph $G = (V, E)$ and let $H = (V, F)$
be a \emph{maximal} eulerian subgraph. Let $r$ be a rank function of $V$
conforming $H$.  Then
\[
	\agony{G, r} = \abs{F} + \sum_{e \in F}  \slack{e, r}\quad.
\]
Moreover, if the sum of slacks is $0$, then $r$ has the lowest possible agony.
\end{proposition}

\begin{proof}
Since $H$ is an eulerian graph, we can partition $H$ into $s$ edge-disjoint
cycles $C_1, \ldots, C_s$.
Since backward edges are only in $H$ we can write agony as
\[
	\agony{G, r} = \sum_{e \in F} \agony{e, r} = \sum_{i = 1}^s \sum_{e \in C_i} \agony{e, r}\quad.
\]
The agony of a single edge $e = (u, v)$ can be written as
\[
\begin{split}
	\agony{e, r} & = \max(r(v) - r(u) + 1, 0) = r(v) - r(u) + 1 - \min(r(v) - r(u) + 1, 0) \\
	& = r(v) - r(u) + 1 + \max(r(u) - r(v) - 1, 0) \\
	& = r(v) - r(u) + 1 + \slack{e, r}\quad.
\end{split}
\]
Summing the edges in a single cycle gives us
\[
	\sum_{e \in C_i} \agony{e, r} = \sum_{e = (u, v) \in C_i} r(v) - r(u) + 1 + \slack{e, r} = \abs{C_i} + \sum_{e \in C_i}\slack{e, r}\quad.
\]
Since the cycles are edge-disjoint, we get the first result of the proposition.
If the sum of slacks is $0$, then the the agony $\agony{G, r}$ is equal to the number
of edges in eulerian subgraph. Proposition~\ref{prop:primaldual} now implies that $r$ is optimal
and $H$ is in fact a maximum eulerian subgraph.
\qed
\end{proof}

\begin{example}
Consider the left graph in Figure~\ref{fig:ex}. The current agony is equal to $8$ and
the size of the current eulerian subgraph is equal to $7$. There is one slack edge, namely $\slack{(a, c), r} = 1$.
On the other hand, the right graph in Figure~\ref{fig:ex} has agony of $8$ which is equivalent
to the number of edges in the eulerian subgraph. There are no slack edges.
\end{example}

\subsection{Algorithm for computing agony}

We are ready to describe the algorithm. Assume that we are given a graph $G$
and assume that we have obtained a maximal eulerian subgraph and a rank $r$
that conforms $H$. We will describe later how to obtain the initial $H$ and
$r$.

Proposition~\ref{prop:slack} states that $r$ is optimal if there are no edges
with slack in $H$. Assume there is one, say $(p, s)$. We begin the algorithm by
increasing the rank of $p$ so that $(p, s)$ has no slack. This may result that
some of the edges outside $H$ become backward, hence we will increase the rank
of the end point of each new backward edge to make sure that there are no new
backward edges. In addition, some of edges $H$ may obtain more slack, hence we
will also increase those vertices. These increases may require additional
increases for other vertices and we keep doing this until either there are no
more increases needed. If we do not encounter $s$ during this algorithm, then
we have successfully reduced agony by the $\slack{(p, s), r}$.  Otherwise, we
will show that we can modify $H$ such that the number of edges in increased.

The visiting order of vertices is important in order to guarantee that the
algorithm runs in $O(m)$ time. We will show that we can guarantee the running
time if we keep the vertices in a priority queue based on how much we need to
increase their rank, larger increases first.

The pseudo-code for the algorithm is given in Algorithm~\ref{alg:relief}.  The
algorithm takes as an input the underlying graph $G$, current maximal eulerian
subgraph $H$ and conforming $r$, and an edge $(p, s) \in E(H)$ with positive
slack.  The algorithm outputs a new subgraph $H'$ and a new rank function $r'$.

\begin{algorithm}[ht!]
\caption{\relief, given an maximal eulerian subgraph $H$ and a conforming rank
$r$, computes a new subgraph $H'$ and a new rank function $r'$ such that the agony or $r$ is closer to the number of edges in the subgraph.}
\label{alg:relief}
\Input{underlying graph $G$, current maximal eulerian subgraph $H$, current rank function $r$, $(p, s) \in E(H)$ an edge with positive slack}
\Output{updated maximal eulerian subgraph and new rank function}

$F \define E(H)$\;
$r' \define r$\;
$t(v) \define 0$ for all $v \in V$ \tcpas{how much we need to increase $v$} 
$t(p) \define r(s) - r(p) - 1$\;
add $p$ to $S$ with priority $t(p)$\;

\While {$S$ is not empty} {
	$u \define $ pop first element from $S$\;

	$r'(u) \define r'(u) + t(u)$\;
	\ForEach{$(u, v) \in E \setminus F$} {
		\If {$r'(v) \leq r'(u)$} {
			$t \define r'(u) + 1 - r'(v)$\;
			\If {$t > t(v)$} {
				$t(v) \define t$\;
				add $v$ to $S$ with priority $t$, update $v$ if $v \in S$ already\;
				$\parent{v} \define u$\;
			}
		}
	}
	\ForEach{$e = (w, u) \in F$} {
		\If {$\slack{e, r'} > \slack{e, r}$} {
			$t \define \slack{e, r'} - \slack{e, r}$\;
			\If {$t > t(w)$} {
				$t(w) \define t$\;
				add $w$ to $S$ with priority $t$, update $w$, if $w \in S$ already\;
				$\parent{w} \define u$\;
			}
		}
	}

	\If {$\slack{(p, s), r'} > 0$} {
		$O \define $ edges in $E$ along the path from $s$ to $p$ using $\parent{}$\; 
		$F \define (F \setminus O) \cup (O \setminus F)$\;
		delete $(p, s)$ from $F$\;
	}

	\Return $(V, F),\  r'$\;
}

\end{algorithm}

\begin{figure}[ht!]

Case 1: we can increase $r(a)$ without increasing $r(c)$\\
\subcaptionbox{input graph\label{fig:toy:a}}{
\begin{tikzpicture}

\draw[leveledge] (-0.5\nodesep, 0) -- (4.5\nodesep, 0);
\draw[leveledge] (-0.5\nodesep, \levelsep) -- (4.5\nodesep, \levelsep);
\draw[leveledge] (-0.5\nodesep, 2\levelsep) -- (4.5\nodesep, 2\levelsep);
\draw[leveledge] (-0.5\nodesep, 3\levelsep) -- (4.5\nodesep, 3\levelsep);

\node[node] (u0) at (\nodesep, 0) {$a$};
\node[node] (u1) at (0, \levelsep) {$b$};
\node[node] (u2) at (\nodesep, 2\levelsep) {$c$};
\node[node] (u3) at (0, 3\levelsep) {$d$};

\node[node] (u4) at (2\nodesep, 1\levelsep) {$e$};
\node[node] (u5) at (3\nodesep, 0\levelsep) {$f$};
\node[node] (u6) at (4\nodesep, 1\levelsep) {$g$};
\node[node] (u7) at (3\nodesep, 2\levelsep) {$h$};

\draw[euleredge] (u0) edge (u2);
\draw[euleredge] (u2) edge (u3);
\draw[euleredge] (u3) edge (u1);
\draw[euleredge] (u1) edge (u0);

\draw[euleredge] (u5) edge (u4);
\draw[euleredge] (u4) edge (u6);
\draw[euleredge] (u6) edge (u5);
\draw[dagedge] (u6) edge (u7);
\draw[dagedge] (u4) edge (u7);

\draw[dagedge] (u0) edge (u4);
\end{tikzpicture}}\hfill
\subcaptionbox{final graph\label{fig:toy:b}}{
\begin{tikzpicture}

\draw[leveledge] (-0.5\nodesep, 0) -- (4.5\nodesep, 0);
\draw[leveledge] (-0.5\nodesep, \levelsep) -- (4.5\nodesep, \levelsep);
\draw[leveledge] (-0.5\nodesep, 2\levelsep) -- (4.5\nodesep, 2\levelsep);
\draw[leveledge] (-0.5\nodesep, 3\levelsep) -- (4.5\nodesep, 3\levelsep);

\node[node] (u0) at (\nodesep, \levelsep) {$a$};
\node[node] (u1) at (0, \levelsep) {$b$};
\node[node] (u2) at (\nodesep, 2\levelsep) {$c$};
\node[node] (u3) at (0, 3\levelsep) {$d$};

\node[node] (u4) at (2\nodesep, 2\levelsep) {$e$};
\node[node] (u5) at (2.5\nodesep, 1\levelsep) {$f$};
\node[node] (u6) at (4\nodesep, 1\levelsep) {$g$};
\node[node] (u7) at (3\nodesep, 3\levelsep) {$h$};

\draw[euleredge] (u0) edge (u2);
\draw[euleredge] (u2) edge (u3);
\draw[euleredge] (u3) edge (u1);
\draw[euleredge] (u1) edge (u0);

\draw[euleredge] (u5) edge (u4);
\draw[euleredge] (u4) edge (u6);
\draw[euleredge] (u6) edge (u5);
\draw[dagedge] (u6) edge (u7);
\draw[dagedge] (u4) edge (u7);

\draw[dagedge] (u0) edge (u4);
\end{tikzpicture}}\hfill\hspace*{5\nodesep}\\[2mm]

Case 2: we cannot increase $r(a)$ without increasing $r(c)$

\subcaptionbox{input graph\label{fig:toy:c}}{
\begin{tikzpicture}

\draw[leveledge] (-0.5\nodesep, 0) -- (4.5\nodesep, 0);
\draw[leveledge] (-0.5\nodesep, \levelsep) -- (4.5\nodesep, \levelsep);
\draw[leveledge] (-0.5\nodesep, 2\levelsep) -- (4.5\nodesep, 2\levelsep);
\draw[leveledge] (-0.5\nodesep, 3\levelsep) -- (4.5\nodesep, 3\levelsep);

\node[node] (u0) at (\nodesep, 0) {$a$};
\node[node] (u1) at (0, \levelsep) {$b$};
\node[node] (u2) at (\nodesep, 2\levelsep) {$c$};
\node[node] (u3) at (0, 3\levelsep) {$d$};

\node[node] (u4) at (2\nodesep, 1\levelsep) {$e$};
\node[node] (u5) at (3\nodesep, 0\levelsep) {$f$};
\node[node] (u6) at (4\nodesep, 1\levelsep) {$g$};
\node[node] (u7) at (3\nodesep, 2\levelsep) {$h$};

\draw[euleredge] (u0) edge (u2);
\draw[euleredge] (u2) edge (u3);
\draw[euleredge] (u3) edge (u1);
\draw[euleredge] (u1) edge (u0);

\draw[euleredge] (u5) edge (u4);
\draw[euleredge] (u4) edge (u6);
\draw[euleredge] (u6) edge (u5);

\draw[dagedge] (u6) edge (u7);
\draw[dagedge] (u4) edge (u7);

\draw[dagedge] (u0) edge (u4);
\draw[dagedge] (u4) edge (u2);
\end{tikzpicture}}\hfill
\subcaptionbox{after increasing ranks\label{fig:toy:d}}{
\begin{tikzpicture}

\draw[leveledge] (-0.5\nodesep, 0) -- (4.5\nodesep, 0);
\draw[leveledge] (-0.5\nodesep, \levelsep) -- (4.5\nodesep, \levelsep);
\draw[leveledge] (-0.5\nodesep, 2\levelsep) -- (4.5\nodesep, 2\levelsep);
\draw[leveledge] (-0.5\nodesep, 3\levelsep) -- (4.5\nodesep, 3\levelsep);

\node[node] (u0) at (\nodesep, \levelsep) {$a$};
\node[node] (u1) at (0, \levelsep) {$b$};
\node[node] (u2) at (\nodesep, 3\levelsep) {$c$};
\node[node] (u3) at (0, 3\levelsep) {$d$};

\node[node] (u4) at (2\nodesep, 2\levelsep) {$e$};
\node[node] (u5) at (2.5\nodesep, 1\levelsep) {$f$};
\node[node] (u6) at (4\nodesep, 1\levelsep) {$g$};
\node[node] (u7) at (3\nodesep, 3\levelsep) {$h$};

\draw[euleredge] (u0) edge (u2);
\draw[euleredge] (u2) edge (u3);
\draw[euleredge] (u3) edge (u1);
\draw[euleredge] (u1) edge (u0);

\draw[euleredge] (u5) edge (u4);
\draw[euleredge] (u4) edge (u6);
\draw[euleredge] (u6) edge (u5);
\draw[dagedge] (u6) edge (u7);
\draw[dagedge] (u4) edge (u7);

\draw[dagedge] (u0) edge (u4);
\draw[dagedge] (u4) edge (u2);
\end{tikzpicture}}\hfill
\subcaptionbox{final graph\label{fig:toy:e}}{
\begin{tikzpicture}

\draw[leveledge] (-0.5\nodesep, 0) -- (4.5\nodesep, 0);
\draw[leveledge] (-0.5\nodesep, \levelsep) -- (4.5\nodesep, \levelsep);
\draw[leveledge] (-0.5\nodesep, 2\levelsep) -- (4.5\nodesep, 2\levelsep);
\draw[leveledge] (-0.5\nodesep, 3\levelsep) -- (4.5\nodesep, 3\levelsep);

\node[node] (u0) at (\nodesep, \levelsep) {$a$};
\node[node] (u1) at (0, \levelsep) {$b$};
\node[node] (u2) at (\nodesep, 3\levelsep) {$c$};
\node[node] (u3) at (0, 3\levelsep) {$d$};

\node[node] (u4) at (2\nodesep, 2\levelsep) {$e$};
\node[node] (u5) at (2.5\nodesep, 1\levelsep) {$f$};
\node[node] (u6) at (4\nodesep, 1\levelsep) {$g$};
\node[node] (u7) at (3\nodesep, 3\levelsep) {$h$};

\draw[dagedge] (u0) edge (u2);
\draw[euleredge] (u2) edge (u3);
\draw[euleredge] (u3) edge (u1);
\draw[euleredge] (u1) edge (u0);

\draw[euleredge] (u5) edge (u4);
\draw[euleredge] (u4) edge (u6);
\draw[euleredge] (u6) edge (u5);
\draw[dagedge] (u6) edge (u7);
\draw[dagedge] (u4) edge (u7);

\draw[euleredge] (u0) edge (u4);
\draw[euleredge] (u4) edge (u2);
\end{tikzpicture}}

\caption{Two examples of applying \relief for $(a, c)$. Dotted edges
represent the eulerian subgraph. Ranks are represented by dashed grey horizontal
lines.}
\label{fig:toy}
\end{figure}

\begin{example}
Consider the graph given in Figure~\ref{fig:toy:a}. The eulerian subgraph is
marked with orange dotted edges and the current rank function is represented by the
dashed grey lines. Edge $(a, c)$ has a slack of $1$. Consider applying \relief on edge $(a, c)$.
The algorithm first increases $r(a)$. Edge $(a, e)$ is no longer a forward edge, hence we need to increase $e$.
This in turns transforms edge $(e, h)$ into backward and increases the slack of $(f, e)$.
Ranks for both vertices are also increased. No other modifications are needed and the final graph is
given in Figure~\ref{fig:toy:b}.

Now consider the graph given in Figure~\ref{fig:toy:c} and apply \relief on edge $(a, c)$.
As in previous case, $e$, $f$, and $h$ are increased, but in addition $c$. Note that
we did not manage to reduce the slack between $a$ and $c$. However, if travel back along
the \parent{} links, $\parent{c} = e$ and $\parent{e} = a$ we obtain a path from $c$ to $a$.
By replacing $(a, c)$ with these edges in the eulerian subgraph we obtain a new subgraph
that has more edges.
\end{example}

The previous example showed the two possible outcomes for \relief, in both
cases we reduce the slack.  The following proposition states that this holds in
general, that is, the new $H$ and $r$ are valid and that the difference
between the costs is smaller.

\begin{proposition}
\label{prop:correct}
Assume that we are given a graph $G$. Let $H = (V, F)$ be a maximal eulerian subgraph of
$G$ and let $r$ be a rank function conforming $H$. Assume that there is an edge
$(p, s) \in E(H)$ such that $\slack{(p, s), r} > 0$. Let $H', r' = \relief(H, r, G, p, s)$.
Then $H'$ is a maximal eulerian subgraph of $G$, $r'$ is conforming $H'$,
$\max_{e \in E(H')} \slack{e, r'} \leq \max_{e \in E(H)} \slack{e, r}$, and 
\[
	\abs{\set{e \in E(H') \mid \slack{e, r'} > 0}} < \abs{\set{e \in E(H) \mid \slack{e, r} > 0}}\quad.
\]
\end{proposition}

In order to prove this result we need the following lemma.

\begin{lemma}
\label{lem:visit}
Each vertex visited at most once during \relief.
Order the visited vertices based on their visiting order, say $u_k$.
Let $t_k = t(u_k)$, where $t(u_k)$ is the priority of $u_k$ at the time when $u_k$
is visited. Then $t_{k + 1} \leq t_k$.
\end{lemma}

\begin{proof}
We will prove by induction over the iteration of \relief that once
a vertex $u$ has been removed from $S$ it will never be added again to $S$
and the priorities of newly added vertices into $S$ during processing $u$
is at most $t(u)$. 

Assume that this holds for $k - 1$ first iterations, and let $u = u_k$ be a vertex
that is visited during the $k$th iteration. Since $S$ selects elements with 
the highest priorities, the induction assumption implies that $t_{i + 1} \leq t_{i}$
for $i = 1, \ldots, k - 2$.

Let $(u, v) \in E \setminus F$. Let $t = r'(u) + 1 - r'(v)$. Since $u$ is
visited for the first time, we must have $r'(u) = r(u) + t(u)$ which implies
that $t = t(u) + r(u) + 1 - r'(v) \leq t(u)$, where the inequality holds since
$(u, v)$ is a forward edge w.r.t. $r$.  If $v$ is added in $S$, then its
priority is at most $t(u)$.  This proves the second part of the induction step.
On the other hand, if $v$ has been already visited, say $v = u_j$, then $t_j = t(v) \geq t_k$ and
$v$ will not be added into $S$.

A similar argument can be made for the edges in $F$.

This proves the induction step and the lemma as the first step is trivial.\qed
\end{proof}

\begin{proof}[of Proposition~\ref{prop:correct}]
Let us consider two separate cases. In Case 1, $s$ remains
unvisited while in Case 2 we visit $s$.

\emph{Case 1:}
Assume that we do not visit $s$.  In such case, $H' = H$, hence $H'$ is maximal.

We need to first show that edges in $E \setminus F$ remain forward.  Whenever
we increase the rank of $u$ we check that none of the edges in $E \setminus F$
are backward.  Assume there is one, say $(u, v)$. If $v$ is already visited,
then Lemma~\ref{lem:visit} states that $t(v) \geq t(u)$.  Since each vertex is
visited only once, this implies that $r'(v) = r(v) + t(v)$ and $r'(u) = r(u) +
t(u)$. This is a contradiction since $(u, v)$ is a forward edge w.r.t. $r$.
Hence, either $v$ is not visited or is in $S$. Either way, we will increase
$r'(v)$ so that $v$ will become a forward edge at some point.

Using similar argument, we see that the slack of edges in $F$ is not increased.
Since the edge $(p, s)$ is no longer a slack edge and we do not increase slack
of any other edges, we have proved the proposition for Case 1.

\emph{Case 2:}
Assume that we have visited $s$. Write $F' = E(H')$.

Let us first argue that we can reach $p$ by using the $\parent{}$ links.
Lemma~\ref{lem:visit} implies that each vertex is visited only once which
guarantees that $\parent{}$ links form a tree whose root is $p$.

Using the same argument as in Case 1, we see that forward edges in $E \setminus
F$ remain forward edges and the slackness of edges in $F$ is not increased.
Moreover, one can easily show that $(p, s)$ and the edges in $O \cap F$
are also forward edges.  This means that
$E \setminus F'$ contains only forward edges. This means that $r'$ conforms
$H'$ and $E \setminus F'$ form a DAG which is only possible when $H'$ is maximal.
It is easy to see that $H'$ is also eulerian.

Let $O_1 = O \cap (E \setminus F)$. 
For any edge
$(u, v) \in O_1$, we must have $r'(v) = r'(u) + 1$,
otherwise $\parent{v} \neq u$. This shows that the slack of the new edges is $0$.
Since $(p, s)$ is removed from $H'$ and we do not increase slack
of any other edges, we have proved the proposition for Case 2.
\qed
\end{proof}

Our next step that this single iteration is linear in the number of edges.

\begin{proposition}
\label{prop:relieftime}
The running time of \relief is $O(m + \slack{(p, s), r})$.
\end{proposition}

\begin{proof}
Since each vertex is visited only once (Lemma~\ref{lem:visit}) we will consider
each edge only once. Hence, the inner for-loops are executed $m$ times at most.
Since the priorities of vertices are integers, we can implement the priority
queue by storing each vertex into an array of $ \slack{(p, s), r} - 1$ linked
lists. Inserting or updating a vertex will take a constant time.  Since the
new priorities will always be smaller or equal, obtaining the maximum element
takes $O(\slack{(p, s), r})$ of \emph{total} time due to the fact that we need
to possibly check some empty linked lists.
This proves the proposition.
\qed
\end{proof}

Alternatively, we can implement the priority queue as a heap which gives the running
time to be $O(m \log n)$.

In practice, we also apply the following speed-up. We monitor $t(s)$ constantly
and we visit only those vertices that have larger priority. Since,
Lemma~\ref{lem:visit} states that the priorities are non-increasing, we simply
stop the main loop once we encounter a vertex with the same priority as $t(s)$.
In addition, once we are done we backtrack rank of each visited vertex by $t(s)$.  If
$s$ is not visited, then $t(s)$ remains $0$ and this speed-up has no effect.
However, if $s$ is inserted in the stack $S$, we will prune vertices that have
the same or lower priority than $S$. We ignore any vertex that should be lowered
by at most $t(s)$. This may transform some forward edges into backward edges but
we counter this by lowering the rank of the already visited vertices by $t(s)$.
This implies that the forward edges remain forward and the arguments done in proof of Proposition~\ref{prop:correct}
are valid.

We are now ready to state the main loop, given in Algorithm~\ref{alg:minagony},
which applies \relief to the edge with the largest slack.

\begin{algorithm}[ht!]
\caption{\minagony, given a graph $G$, a maximal eulerian subgraph $H$ and a
rank function $r$ conforming $H$, finds a rank function  optimizing agony}
\label{alg:minagony}
\Input{underlying graph $G$, maximal eulerian subgraph $H$, rank function $r$}
\Output{optimal maximal eulerian subgraph and rank function}

\While {$\agony{G, r} > \abs{E(H)}$} {
	$(p, s) \define $ an edge in $E(H)$ with largest slack\;
	$H, r \define \relief(H, r, G, p, s)$\;

}

\Return $H,\ r$\;
\end{algorithm}

Our next step is to show that we need to call \relief at most $O(m)$. 
\begin{proposition}
\label{prop:minagonytime}
Assume a graph $G$, a maximal eulerian subgraph $H$ and a rank function conforming $H$
such that $\slack{e, r} \leq m$ for any edge $e \in H$. Then $\minagony(G, H, r)$ takes $O(m^2)$ time.
\end{proposition}

\begin{proof}
Proposition~\ref{prop:correct} states that each call reduces the number of
slack edges by at least $1$. There can be at most $m$ slack edges.  Hence, the
number of \relief calls is at most $m$.  Since $\slack{e, r} \leq m$ at the
beginning and Proposition~\ref{prop:correct} states that slack is never
increased, Proposition~\ref{prop:relieftime} implies that calling \relief 
takes $O(m)$ time. This completes the proof.
\qed
\end{proof}

Assume that we are given $H$, a maximal eulerian subgraph of $G$.  Then $E(G)
\setminus E(H)$ is a DAG, and any topological order will provide a rank
function that is conforming with $H$. In this paper, we use a rank function,
where we first remove all source vertices simultaneously from the DAG and
assign them the same rank. We continue this until DAG is empty.
The largest rank in this case is at most $n$, this also bounds the slack
and consequently the conditions in Proposition~\ref{prop:minagonytime} are satisfied. 

\subsection{Discovering maximal eulerian subgraph}

Our final step is to discover a maximal eulerian subgraph.  This can be done
naively by running a DFS, finding and a removing a cycle and repeating until no
cycles are left. This gives us running time of $O(m^2)$. A more sophisticated
approach can be done with a single DFS, given in Algorithm~\ref{alg:dfs}.

\begin{algorithm}[ht!]
\caption{\cycledfs, discovers a maximal eulerian subgraph.}
\label{alg:dfs}
\Input{$G$, directed graph}
\Output{$F$, edges corresponding to a maximal eulerian subgraph}

\While {$V \neq \emptyset$} {
	$S \define $ any vertex in $V$\;
	\While {$S \neq \emptyset$} {
		$u \define$ first vertex in $S$\;
		\If {there is $(u, v) \in E$} {
			\If {$v \in S$} { 
				$O \define (u, v)$ and the path from $v$ to $u$ along $S$\;
				$F \define F \cup O$\;
				delete $O$ from $G$\;
				pop vertices from $S$ until the last vertex is $v$\;
			}
			\Else {
				push $v$ to $S$\;
			}
		}
		\Else {
			pop $u$ from $S$\;
			remove $u$ from $G$\;
		}
	}
}
\Return $F$\;
\end{algorithm}

\cycledfs starts with DFS and the moment it discovers a back edge, it finds a corresponding cycle.
The algorithm proceeds by deleting the cycle and backtracking to the first vertex of visited cycle. 
The following proposition shows that the algorithm indeed finds a maximal eulerian subgraph. 

\begin{proposition}
\cycledfs discovers maximal eulerian subgraph.
\end{proposition}

\begin{proof}
$F$ consists of edge-disjoint cycles, and by definition is eulerian.
Assume that $F$ is not maximal, that is, there is a cycle $C$.
Let $u$ be the first vertex in $C$ that is deleted from $G$. 
Let $e$ be the outgoing edge from $u$ in $C$. By definition, $e$ is not added
in $F$. This implies that when we delete $u$ from $G$, $e$ is still present in
$G$ which is a contradiction since we only delete vertices with no outgoing edges.
\qed
\end{proof}

As a final step we show that \cycledfs runs in linear time.

\begin{proposition}
\cycledfs executes in $O(m)$ time.
\end{proposition}

\begin{proof}
During a single iteration of the inner while-loop we either delete $x$ edges or
push a vertex into a stack. Hence, the total running time is bounded by the
number of edges deleted plus the number of pushes. Since each edge can be
deleted only once, the first term is bounded by $m$. The number of times
we will push a vertex $u$ into $S$ is bounded by the in-degree of $u$ plus 1.
Consequently, the number of pushes we will do in total is $O(n + m)$, which proves the result.
\qed
\end{proof}

\section{Related work}
\label{sec:related}

From algorithmic point of view, the relation between our approach and the
algorithm given by~\citet{gupte:11:agony} is intriguing.  Both methods are
based on primal-dual techniques, that is, they rely on the relationship between
the primal problem, minimizing agony, and the dual problem, maximizing the
eulerian subgraph. Gupte's algorithm is essentially an instance of the
primal-dual  algorithm, where one tries to improve the dual problem, in this
case discovering maximum eulerian subgraph, until no improvement is possible.
This improvement correspond to finding the negative cycle in a certain weighted
graph, that is, a cycle whose sum of weights is negative.
Currently the best algorithm for discovering negative cycle needs $O(nm)$
time~\cite{DBLP:conf/esa/CherkasskyG96} and this can be achieved with a Bellman-Ford algorithm~\cite{cormen:01:algorithm}.
Since we need $m$ iterations at most, the computational complexity of this approach is $O(nm^2)$.

On the other hand, our approach is also an instance  of the primal-dual
algorithm. Especially, both algorithms improve the current eulerian subgraph.
The difference is that while Gupte's algorithm searches the improvement by
transforming the problem into discovering negative cycles, we discover the
improvement in several calls of \relief. During each call of \relief if we have
not able to find a new improvement for the eulerian subgraph, then we are able
to improve the primal problem, that is, minimizing agony. In other words, while
searching for improvement for the eulerian subgraph, we are able to use
intermediate calculations to minimize the agony. This allows us to achieve
a better computational complexity of $O(m^2)$. 

The agony of a single edge is chosen very carefully. For example, if we choose
agony to be $1$ for every backward edge, then the problem is related to
\textsc{Feedback Arc Set}, where the goal is to discover a directed acyclic graph $H$.
from a given directed graph $G$ such that $E(G) \setminus E(H)$ is minimized.
This problem is not only \np-hard, it is also \apx-hard with a coefficient of $c = 1.3606$~\cite{dinur:05:cover}.
There is no known constant-ratio approximation algorithm for \textsc{FAS} and
the best known approximation algorithm has ratio $O(\log n \log \log n)$~\cite{even:98:feedback}.

Next, we highlight some of the existing methods for discovering hierarchies.
\citet{DBLP:conf/cse/MaiyaB09} suggested a statistical model where the
probability of an edge is high between a parent and a child. To find the
hierarchy they employ a greedy heuristic. \citet{clauset:08:hierarchy} studied
discovering hierarchy in undirected graphs, where given a dendrogram, the
probability of an edge between two vertices is based on Erd\H{o}s-R\'{e}nyi
model, with a probability depending on the lowest common ancestor in the
dendrogram. The authors then sample dendrograms using MCMC techniques.
\citet{DBLP:conf/icdm/MacchiaBGC13} used agony to discover summaries of
propagations based on traces. \citet{jameson:99:behaviour} applied a model,
where the likelihood of the the vertex dominating other is based on the
difference of their ranks, to animal dominance data. Similar ideas has been used
for ranking chess players by~\citet{elo1978rating}.
Finally, hierarchy partitions vertices into groups, the top-level vertices having very
different role than the bottom-level vertices. Assigning different roles to
vertices have received some attention. \citet{henderson:12:roix} consider assigning roles
to vertices based on features while \citet{mccallum:07:roles} assigned topic
distributions to individual vertices. An interesting direction for future work
would be to study how hierarchy can be used for role mining in graphs.

\section{Experimental evaluation}\label{sec:exps}

While we were able to improve the computational complexity of computing agony
from $O(nm^2)$ to $O(m^2)$, the bound is still impractical even for graphs of
modest size. Our next goal is to demonstrate empirically that this bound is in
fact pessimistic and that we can compute the agony for large graphs.

In order to do so, we applied our algorithm for several large directed graphs,
downloaded from Stanford Large Network Dataset Collection
(SNAP).\!\footnote{The datasets and their detailed descriptions are available
at http://snap.stanford.edu/data/index.html} We removed any edges of form $(u,
u)$ as they have no effect on the rank.  The characteristics of the datasets
are given in the first 2 columns in Table~\ref{tab:basic}.
In addition, to our algorithm we applied a baseline algorithm of~\citet{gupte:11:agony}. The algorithm
requires a subroutine for detecting a negative cycle. We used Bellman-Ford
algorithm with an additional speed-up, where after each iteration over the
edges we check whether a cycle has been discovered.
We implement both algorithms in C++ and performed experiments using a Linux-desktop  
equipped with a Opteron 2220 SE processor.
The running times and detailed statistics are given in Table~\ref{tab:basic}.

\begin{table}[ht!]
\caption{Basic characteristics of the datasets and statistics from experiments.
The 3rd column indicates the number of iterations, the 4th column indicates
the slack of the starting point, the 5th depicts the final score.
The running times for \minagony and the baseline are given in 6th and 7th columns, respectively.}
\label{tab:basic}
\begin{tabular*}{\textwidth}{@{\extracolsep{\fill}}lrrrrrrr}
\toprule
Dataset & $\abs{V}$ & $\abs{E}$ & iterations & gap & agony & time & baseline \\
\midrule
Amazon     & 403\,394 & 3\,387\,388 &  89\,046 &    911\,095 & 1\,973\,965 & 4h27m & --\\
Gnutella   &  62\,586 &    147\,892 &   1\,907 &    150\,851 &     18\,964 & 45s   & 20m\\
EmailEU    & 265\,214 &    418\,956 &  27\,679 &    500\,177 &    120\,874 & 2m    & 3h45m\\
Epinions   &  75\,879 &    508\,837 &  18\,652 &    922\,817 &    264\,995 & 20m   & 1h40m\\
Slashdot   &  82\,168 &    870\,161 &  37\,858 & 1\,891\,586 &    748\,582 & 1h5m  & 7h3m\\
WebGoogle  & 875\,713 & 5\,105\,039 & 164\,708 & 4\,110\,696 & 1\,841\,215 & 2h32m & --\\
WikiVote   &   7\,115 &    103\,689 &      865 &     76\,149 &     17\,676 & 7s    & 1m\\
\bottomrule
\end{tabular*}
\end{table}

Our first observation is that the theoretical bound is indeed pessimistic.  We
are able to compute the agony for large networks in reasonable time. We spend 2.5
hours for computing agony for \emph{WebGoogle}, a graph with over 5 million
edges, and 4 hours for \emph{Amazon}, a graph with over 3 million edges.  For
smaller graphs, the running time is significantly faster, either seconds or
minutes.

The reason for this scalability is two-fold. First of all, the number of
iterations, given in 4th column, is significantly lower than the number of
edges. Secondly, \relief typically affects less than $m$ vertices.

Our method
performs better than the baseline for all datasets. We interrupted the baseline
calculation after 24 hours for \emph{Amazon} and \emph{WebGoogle}.

Finally, let us consider behaviour of agony and the size of the eulerian
graph as a function of iterations. In order to do so we plot the
evolution of scores, normalized by the final agony, in Figure~\ref{fig:anytime}.

\begin{figure}[ht!]
\newlength{\imgwidth}
\setlength{\imgwidth}{2.25cm}

￼\pgfplotsset{
    xtick scale label code/.code={}
}
\begin{tikzpicture}
\begin{axis}[xlabel={iterations $\times 10^{4}$},ylabel= {score / final},
    width = \imgwidth,
    cycle list name=yaf,
	scale only axis,
	title = {\emph{EmailEU}},
	no markers
    ]
\addplot table[x index = 0, y index = 1, header = false] {results/email-EuAll.dat};
\addplot table[x index = 0, y index = 2, header = false] {results/email-EuAll.dat};
\pgfplotsextra{\yafdrawaxis{0}{27600}{5.102}{0.964}}
\end{axis}
\end{tikzpicture}
\begin{tikzpicture}
\begin{axis}[xlabel={iterations $\times 10^{3}$},
    width = \imgwidth,
    cycle list name=yaf,
	scaled x ticks = base 10:-3,
	scale only axis,
	title = {\emph{Gnutella}},
	ytick = {1, 3, 5, 7},
	no markers
    ]
\addplot table[x index = 0, y index = 1, header = false] {results/p2p-Gnutella31.dat};
\addplot table[x index = 0, y index = 2, header = false] {results/p2p-Gnutella31.dat};
\pgfplotsextra{\yafdrawaxis{0}{1899}{ 8.587}{0.650}}
\end{axis}
\end{tikzpicture}
\begin{tikzpicture}
\begin{axis}[xlabel={iterations $\times 10^{4}$},
    width = \imgwidth,
    cycle list name=yaf,
	scale only axis,
	title = {\emph{Epinions}},
	no markers
    ]
\addplot table[x index = 0, y index = 1, header = false] {results/soc-Epinions1.dat};
\addplot table[x index = 0, y index = 2, header = false] {results/soc-Epinions1.dat};
\pgfplotsextra{\yafdrawaxis{0}{18600}{4.381}{0.899}}
\end{axis}
\end{tikzpicture}
\begin{tikzpicture}
\begin{axis}[xlabel={iterations $\times 10^{4}$},
    width = \imgwidth,
    cycle list name=yaf,
	scale only axis,
	title = {\emph{Slashdot}},
	no markers
    ]
\addplot table[x index = 0, y index = 1, header = false] {results/soc-Slashdot0902.dat};
\addplot table[x index = 0, y index = 2, header = false] {results/soc-Slashdot0902.dat};
\pgfplotsextra{\yafdrawaxis{0}{37800}{3.488}{0.961}}
\end{axis}
\end{tikzpicture}

\begin{tikzpicture}
\begin{axis}[xlabel={iterations $\times 10^{4}$},ylabel= {score / final},
    width = \imgwidth,
    cycle list name=yaf,
	scale only axis,
	title = {\emph{Amazon}},
	no markers
    ]
\addplot table[x index = 0, y index = 1, header = false] {results/amazon0601.dat};
\addplot table[x index = 0, y index = 2, header = false] {results/amazon0601.dat};
\pgfplotsextra{\yafdrawaxis{0}{89000}{1.416}{0.955}}
\end{axis}
\end{tikzpicture}
\begin{tikzpicture}
\begin{axis}[xlabel={iterations $\times 10^{5}$},
    width = \imgwidth,
    cycle list name=yaf,
	scale only axis,
	title = {\emph{WebGoogle}},
	no markers
    ]
\addplot table[x index = 0, y index = 1, header = false] {results/web-Google.dat};
\addplot table[x index = 0, y index = 2, header = false] {results/web-Google.dat};
\pgfplotsextra{\yafdrawaxis{0}{157193}{3.139}{0.917}}
\end{axis}
\end{tikzpicture}
\begin{tikzpicture}
\begin{axis}[xlabel={iterations $\times 10^2$},
    width = \imgwidth,
    cycle list name=yaf,
	scale only axis,
	title = {\emph{WikiVote}},
	scaled x ticks = base 10:-2,
	no markers,
	legend pos = outer north east,
	legend entries = {$\agony{G, r}$, $\abs{E(H)}$}
    ]
\addplot table[x index = 0, y index = 1, header = false] {results/wiki-Vote.dat};
\addplot table[x index = 0, y index = 2, header = false] {results/wiki-Vote.dat};
\pgfplotsextra{\yafdrawaxis{0}{864}{5.142}{0.838}}
\end{axis}
\end{tikzpicture}
\caption{Scores as a function of iteration. Each plot represents a single dataset.
The upper line depicts the current agony score, normalized by the final score, as a function of a current iteration.
The lower line depicts the current number of edges
in the eulerian subgraph, normalized by the final score, as a function of a current iteration. Note that the x-axis is scaled.}
\label{fig:anytime}
\end{figure}
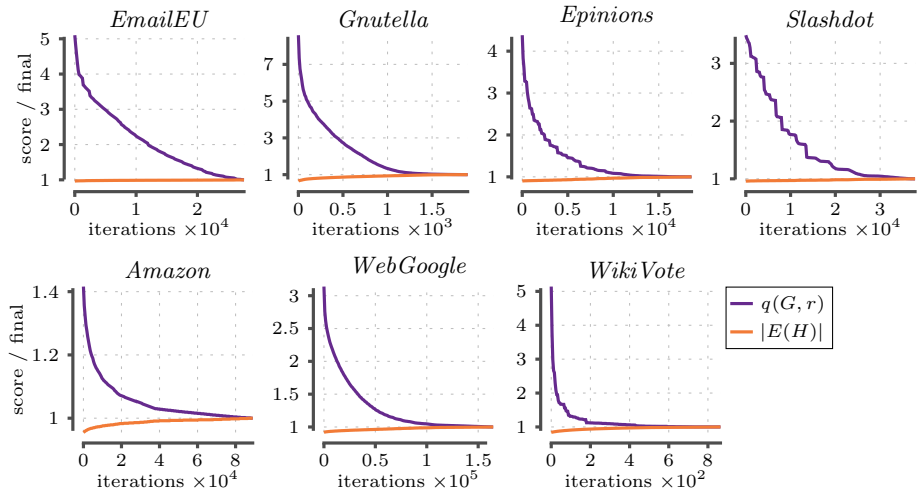

We see that the initial agony is significantly larger than the final
agony. For most datasets the agony drops quickly. For the largest datasets
the algorithm achieves approximation ratio of $2$ relatively quickly:
for \emph{WebGoogle} the algorithm achieves approximation ratio of $2$ during the first 8\%
of iterations. This suggests that we can use the algorithm as any-time algorithm,
stopping iterations early once we achieved acceptable approximation ratio. Note that
since the optimal solution is at least as large as the current eulerian subgraph, we can
at any time bound the approximation ratio of the current agony.

\section{Concluding remarks}\label{sec:concl}

In this paper we introduced an algorithm for discovering hierarchy among
vertices in a given directed graph. The hierarchy should minimize agony, the
edges that violate the hierarchical structure. We show that our algorithm
achieves computational complexity of $O(m^2)$ which is significantly better
than the current bound of $O(nm^2)$. We also demonstrate
that $O(m^2)$ is a pessimistic estimate of the running time and in practice the algorithm
scales up for large networks.

There are several interesting directions for future work. An obvious and
practical extension is to make edges weighted. Weighting edges will change the
definition of the dual problem as we no longer are looking for maximum eulerian
subgraph. On the other hand, integer weights can be viewed as multiple edges
which should imply that the same framework can be applied. Another 
fruitful direction is to consider discovering hierarchies with constraints, such
as the number of hierarchies or demanding that certain vertices have fixed ranks.

\spara{Acknowledgements.} This work was supported by Academy of Finland grant 118653 ({\sc algodan})

\renewcommand\bibname{References}
\bibliography{bibliography}

\end{document}